\documentclass[a4paper]{article}

\usepackage{microtype}
\usepackage{algorithm}

\usepackage{amssymb}
\usepackage[tbtags,fleqn]{amsmath}
\usepackage{enumerate}
\usepackage{amsthm}
\usepackage{array}
\usepackage{multirow}
\usepackage{tabularx}
\usepackage[]{graphicx}

\theoremstyle{plain}
\newtheorem{theorem}{Theorem}
\newtheorem{definition}[theorem]{Definition}
\newtheorem{lemma}[theorem]{Lemma}

\title{Neighborhood Mutual Remainder: Self-Stabilizing Implementation of Look-Compute-Move Robots\protect\linebreak \rm{(Extended Abstract)}}

\author{
Shlomi Dolev\thanks{Department of Computer Science, Ben-Gurion University of the Negev, Israel}
\and
Sayaka Kamei\thanks{Department of Information Engineering, Graduate School of Engineering, Hiroshima University, Japan}
\and
Yoshiaki Katayama\thanks{Department of Computer Science and Engineering, Graduate School of Engineering, Nagoya Institute of Technology, Japan}
\and
Fukuhito Ooshita\thanks{Graduate School of Science and Technology, Nara Institute of Science and Technology, Japan}
\and
Koichi Wada\thanks{Department of Applied Informatics, Faculty of Science and Engineering, Hosei University, Japan}~~\thanks{This author was supported in part by JSPS KAKENHI No. 17K00019 in this research.}
 }

\date{}

\begin{document}

\maketitle

\begin{abstract}
Local mutual exclusion guarantees that no two neighboring processes enter a critical section at the same time while satisfying both \emph{mutual exclusion} and \emph{no starvation} properties. On the other hand, processes may want to execute some operation 
simultaneously with the neighbors. Of course, we can use a globally synchronized clock to achieve the task but it is very expensive to realize it in a distributed system in general. 

In this paper, we define a new concept \emph{neighborhood mutual remainder}. A distributed algorithm that satisfies the neighborhood mutual remainder requirement should satisfy \emph{global fairness}, \emph{$l$-exclusion} and \emph{repeated local rendezvous} requirements. Global fairness is satisfied when each process
(that requests to enter the critical section infinitely often) executes the critical section infinitely often, $l$-exclusion is satisfied when at most $l$ neighboring processes enter the critical section at the same time, and repeated local rendezvous is satisfied when for each process
infinitely often no process in the closed neighborhood is in the critical or trying section. 

We first formalize the concept of neighborhood mutual remainder, and give a simple self-stabilizing algorithm to demonstrate the design paradigm to achieve neighborhood mutual remainder.
We also present two applications of neighborhood mutual remainder to a \emph{Look-Compute-Move} robot system. One is for implementing a move-atomic property and the other is for implementing FSYNC scheduler, where robots possess an independent clock that is advanced in the same speed. These are the first self-stabilizing implementations of the LCM synchronization.
\end{abstract}


\section{Introduction}

Distributed systems sometimes encounter \emph{mutually exclusive} operations such that, while one operation is executed by a participant, another operation cannot be executed by the participant and its neighboring participants. 
For example, consider a database shared by multiple processes. An administrator may backup the database (i.e., execute a backup operation) while no process accesses the database (i.e., executes an access operation). In this case, the backup operation and the access operation are mutually exclusive.

As another example, we can consider a LOOK-COMPUTE-MOVE (LCM) robot system~\cite{FPS}, where each robot repeats executing cycles of LOOK, COMPUTE, and MOVE phases. Some algorithms in the LCM robot system assume the move-atomic property, that is, while robot $r$ executes LOOK and COMPUTE phases, $r$'s neighbors (\emph{i.e.}, robots in $r$'s sight) cannot execute a MOVE phase. In this case, the MOVE operation and the LOOK/COMPUTE operations are mutually exclusive.

To execute mutually exclusive operations in a consistent manner, participants should schedule the operations carefully. One may think we can apply \emph{mutual exclusion} \cite{L96,AW04} or \emph{local mutual exclusion} \cite{AS98,KY02} to solve the local synchronization problem. Mutual exclusion (resp., local mutual exclusion) guarantees that no two participants (resp., no two neighboring participants) enter a critical section at the same time. Indeed, if participants execute mutually exclusive operations only when they are in the critical section, they can keep the consistency because no two neighboring participants execute the mutually exclusive operations at the same time. On the other hand, this approach seems expensive because participants execute the operations sequentially despite that they are allowed to execute the same operation at the same time. In addition, to realize local mutual exclusion, participants should achieve symmetry breaking because one participant should be selected to enter the critical section. However, in highly-symmetric distributed systems such as the LCM robot system, it is difficult or even impossible to achieve deterministic symmetry breaking and thus achieve local mutual exclusion. 

From this motivation, we define the \emph{neighborhood mutual remainder} distributed task over a distributed system with a general, non-necessarily complete, communication graph. A distributed algorithm that satisfies the neighborhood mutual remainder requirement 
should satisfy \emph{global fairness}, \emph{$l$-exclusion}, and \emph{repeated local rendezvous} (or as we use in the sequel for the sake of readability, simply, local rendezvous) requirements. Global fairness is satisfied when each participant executes a critical section 
infinitely often, $l$-exclusion is satisfied when at most $l$ neighboring processes enter the critical section at the same time, and local rendezvous is satisfied when for each participant infinitely often no participant in the closed neighborhood is in the critical section.

Unlike the classical (local) mutual exclusion problem \cite{L96,AW04}
the neighborhood mutual remainder allows (up to $l$, in the sequel we use the number of neighbors plus 1 to be $l$) neighboring participants to be simultaneously in the critical section, but requires a guarantee for neighborhood rendezvous \cite{FSVY16} in the remainder, namely, the state which is not in the critical or trying sections.

As an application example, consider a LOOK-COMPUTE-MOVE (LCM) robot system again. The aforementioned move-atomic property can be achieved by neighborhood mutual remainder: Each robot executes a MOVE phase only when no robot in its closed neighborhood is in the critical section, and executes LOOK and COMPUTE phases only when it is in the critical section. Clearly, while some robot executes LOOK and COMPUTE phases, none of its neighbors executes a MOVE phase. From the global fairness and local rendezvous properties, all robots execute LOOK, COMPUTE, and MOVE phases infinitely often.

One may depict the \emph{neighborhood mutual remainder} distributed task in terms of dynamic graph coloring where each participant should be \emph{red} infinitely often, and at the same time each neighborhood should have infinite instances in which no participant in the closed neighborhood is \emph{red}.
  
One obvious solution to the problem is implied by a distributed synchronizer, as described in \cite{L96,AW04} see \cite{D2K} for \emph{self-stabilizing} synchronizers. The synchronizer may color the system with several colors one of which is \emph{red}, and in this way \emph{all} neighborhood will be non \emph{red} at the same time while each (in fact all simultaneously) will be \emph{red} infinitely often.
One of our solutions is based on such approach in semi-synchronous settings, however as is always the case, global synchronization implies the need to wait for the slowest participant/neighborhood and is less robust to temporal local faults, for example a single participant that stops its operation, say while being \emph{red} can stop the progress of the entire system. 
Hence, local solutions implied by the only neighborhood restriction may be preferred. 

\medskip

\noindent
{\bf Our Contributions.} We first formalize the concept of neighborhood mutual remainder, and give a design paradigm to achieve neighborhood mutual remainder. To demonstrate the design paradigm, we consider synchronous distributed systems and give a simple self-stabilizing algorithm for neighborhood mutual remainder. To simplify the discussion, we assume $l=\Delta+1$, where $\Delta$ is the maximum degree, that is, $l$-exclusion is always satisfied.

After that, to demonstrate applicability of neighborhood mutual remainder, we implement a self-stabilizing synchronization algorithm for an LCM robot system by using the aforementioned design paradigm. As described above, in the LCM robot system, each robot repeats executing cycles of LOOK, COMPUTE, and MOVE phases. First, we realize the move-atomic property in a self-stabilizing manner on the assumption that robots repeatedly receive clock pulses at the same time, where the move-atomic property guarantees that, while some robot executes LOOK and COMPUTE phases, no robot in its sight executes a MOVE phase. After that, we extend the self-stabilizing algorithm to the assumption that robots receive clock pulses at different times but the duration between two pulses is identical for all robots. Lastly, on the same assumption, we implement the FSYNC model in a self-stabilizing manner, that is, based on such an individual clock pulse assumption, we make all robots simultaneously execute LOOK, COMPUTE, and MOVE phases. This research presents the first self-stabilizing implementation of the LCM synchronization, allowing the implementation in practice of any self-stabilizing or stateless robot algorithm, where robots possess independent clocks that are advanced in the same speed.

\medskip

\noindent
{\bf Related works.} For global mutual exclusion problem,
much research has been devoted 
to self-stabilizing algorithms, e.g., \cite{israeli90b} and \cite{KY02J}.
Self-stabilizing distributed algorithms for 
the local mutual exclusion problem are proposed 
in \cite{Beauquier2000} and \cite{KY02}.
Various generalized versions of mutual exclusion have been studied extensively, e.g., $l$-mutual exclusion \cite{ADHK01}\cite{flatebo1995}, mutual inclusion \cite{MUTIN}, $l$-mutual inclusion \cite{MUTIN}, critical section problem \cite{PDAA2018}.

Robots with globally observed light were introduced in~\cite{DFP+16} and used to synchronize the LCM schedules among the robots. In~\cite{DFP+16},  the authors show that asynchronous robots with lights can simulate any algorithm on semi-synchronous robots without robots and thus the asynchronous robots with lights has the same power as the semi-synchronous with lights.
However, unlike our setting, this simulation algorithm is performed asynchronously on the same LCM robot system as the simulated semi-synchronous algorithm works. On the other hand, in our setting, as an application of newly introduced neighborhood mutual remainder, robots utilizing lights and global pulse can implement some LCM schedules such as asynchronous move-atomic and fully-synchronous ones in self-stabilizing manners. Also although in~\cite{DFP+16} unlimited visibility is assumed, in our setting, 
limited visibility is assumed and only neighboring robots observe the light.

\medskip

\noindent
{\bf Roadmap.}
This paper is organized as follows:
Section 2 defines the concept of neighborhood mutual remainder and demonstrates the design paradigm.
Section 3 gives several definitions for robot systems.
Sections 4 and 5 present self-stabilizing move-atomic algorithms with and without global pulses, respectively.
Section 6 presents a self-stabilizing implementation of FSYNC model.
Some details and proofs are omitted from this extended abstract.
Section 7 concludes this paper.

\section{Neighborhood Mutual Remainder}
\label{sec:nmr}

In this section, we introduce a concept of neighborhood mutual remainder and give a design paradigm to achieve neighborhood mutual remainder. To explain the design paradigm in a simple way, we consider fully-synchronous distributed systems and present a self-stabilizing algorithm as an example.

\subsection{A system model}

A distributed system is represented by an undirected connected graph $G=(V,E)$, where $V=\{v_0,\ldots,v_{k-1}\}$ is a set of processes and $E$ is a set of communication links between processes. Processes are anonymous and identical, that is, they have no unique identifiers and execute the same deterministic algorithm. Process $v_i$ is a neighbor of $v_j$ if $(v_i,v_j)\in E$ holds. A neighborhood of $v_i$ is denoted by $N(i)=\{v_j \mid (v_i,v_j) \in E\}$, and the degree of $v_i$ is denoted by $\delta(i)=|N(i)|$. Let $\Delta=\max\{\delta(i) \mid v_i \in V\}$. A closed neighborhood of $v_i$ is denoted by $N[i]=N(i) \cup \{v_i\}$. 

Each process is a state machine that changes its state by actions. We consider the \emph{state-reading model} as a communication model. In this model, each process $v_i$ can directly read a state of $v_j \in N[i]$ and update its own state.

Processes operate synchronously based on global pulses. That is, all processes regularly receive the pulse at the same time, and operate when they receive the pulse. The duration of local computation (including updates of its state) is sufficiently small so that every process completes the local computation before the next pulse.

\subsection{Concept of neighborhood mutual remainder}

In this subsection, we introduce a concept of neighborhood mutual remainder. Analogous to the mutual exclusion task, processes have a critical section in their program, a section that they should enter on a will and must exit thereafter to the remainder section. Unlike mutual exclusion, we require that \emph{all} processes in a closed neighborhood may be infinitely often simultaneously in the remainder section for a while, while having the opportunity to execute the critical section, possibly simultaneously with others, infinitely often too, exiting the critical section following each such entry.

\begin{definition}
\textbf{(Neighborhood mutual remainder)} The system achieves neighborhood mutual remainder if the following three properties hold.
\begin{itemize}
\item Global fairness: Every process infinitely often enters the critical section\footnote{Alternatively global non starvation, where every process willing to enter the critical section infinitely often enters the critical section infinitely often.}.
\item $l$-exclusion: For every process $v_i$, at most $l$ processes in $N[i]$ enter the critical section at the same time.
\item Local rendezvous: For every process $v_i$, infinitely many instants exist such that no process in $N[i]$ is in the critical section or trying section  (i.e., every process in $N[i]$ is in the remainder section)\footnote{In fact, one can define $m$-rendezvous, where $m$ is the subset of neighbors that should be simultaneously in the reminder, defining $(l,m)$-neighboring mutual remainder. In our case, $m$ is the number of neighbors plus 1.}.
\end{itemize}
\end{definition}

\subsection{A self-stabilizing algorithm for neighborhood mutual remainder}

In this subsection, we give a design paradigm to achieve neighborhood mutual remainder. As an example, we realize a self-stabilizing algorithm to achieve neighborhood mutual remainder. To simplify the discussion, we assume $l=\Delta+1$, that is, $l$-exclusion is always satisfied.

\begin{definition}
\textbf{(Self-stabilization)} The system is self-stabilizing if both of the following properties hold.
\begin{itemize}
\item Convergence: The system eventually reaches a desired behavior from any initial configuration, where a configuration is a collection of states of all processes in the system.
\item Closure: Once the system reaches a desired behavior, it keeps the desired behavior after that.
\end{itemize}
\end{definition}

First we give the underlying idea of the self-stabilizing algorithm. Let us consider a simple setting where $|N[i]|$ is identical for any $v_i$. Every process $v_i$ maintains a clock $\mathit{Clock_i}$ that is incremented by 1 modulo $(|N[i]|+1)$ in every pulse. The value of $\mathit{Clock_i}$ may differ from the value of $\mathit{Clock_j}$, for a neighbor $v_j$ of $v_i$. Say, for the sake of simplicity, that $v_i$ may possess the critical section only when $\mathit{Clock_i} = 1$. Thus, ensuring that there is a configuration in which all processes in the remainder is equivalent to ensuring that there is a configuration in which the values of all the above clocks are not equal to 1. Using the pigeon holes principle in every $|N[i]|+1$ consequence pulse clocks, there must be a configuration in which no clock value of a neighboring processes is 1 and at the same time $\mathit{Clock_i}$ is not 1 too. Hence, the neighborhood mutual remainder must hold. 

Since $|N[i]| \neq |N[j]|$ may hold for some $v_i$ and $v_j$, we use $\mathit{MaxN_i}=\max\{|N[j]| \mid v_j \in N[i]\}$ instead of $|N[i]|$. Since every process $v_j\in N[i]$ enters a critical section at most once in $\mathit{MaxN_i}+1$ consecutive pulses, we can still use the pigeon holes principle and hence the neighborhood mutual remainder must hold.

Algorithm~\ref{ssnmr} gives a self-stabilizing algorithm to achieve neighborhood mutual remainder. 
In addition to $\mathit{Clock_i}$, each process $v_i$ has two variables $N_i$ and $\mathit{MaxN_i}$. Process $v_i$ broadcasts $|N[i]|$ to its neighbors by using variable $N_i$. After that $v_i$ computes $\mathit{MaxN_i}$ from $N_j$ ($v_j\in N[i]$) and stores it to $\mathit{MaxN_i}$. Process $v_i$ increments $\mathit{Clock_i}$ modulo $(\mathit{MaxN_i}+1)$, and enters the critical section if $\mathit{Clock_i}=1$.
Process $v_i$ also exposes the value of $\mathit{Clock_i}$ to its neighbors. Thus, a process can rendezvous when all the neighborhood clocks are not equal to 1.

\newcount\lno
\def\XLN{\lno=0}
\def\LN{%
\global\advance\lno by 1%
\hbox to1.2em{\textit{\scriptsize\number\lno:\hfil}}}
\begin{algorithm}[h] 
\begin{tabbing}
10: \= sp \= sp \= sp \= sp \= \kill
\XLN
\LN \> \textbf{Upon} a global pulse\\
\LN \> \> $N_i := |N[i]|$ \\
\LN \> \> $\mathit{MaxN}_i := \max\{ N_j \mid v_j \in N[i]\}$\\
\LN \> \> $\mathit{Clock}_i := (\mathit{Clock}_i+1) \bmod (\mathit{MaxN}_i+1)$ \\
\LN \> \> \textbf{if} $\mathit{Clock_i} = 1$ \textbf{then}\\
\LN \> \> \> Enter the critical section and leave before the next pulse\\
\LN \> \> \textbf{else}\\
\LN \> \> \> // Stay in the remainder section\\
\LN \> \> \> Rendezvous when all neighboring clocks $\neq 1$
\end{tabbing}
\vspace{-0.3cm}
\caption{Self-Stabilizing Neighborhood Mutual Remainder Algorithm for $l=\Delta+1$. Pseudo-Code for $v_i$.}\label{ssnmr}
\end{algorithm}

\begin{theorem}
Algorithm~\ref{ssnmr} achieves neighborhood mutual remainder with $l=\Delta+1$ in a self-stabilizing manner.
\end{theorem}

\begin{proof}
Every process $v_i$ correctly assigns $|N[i]|$ to $N_i$ at the first pulse, and hence it correctly assigns $\max\{|N[j]| \mid v_j\in N[i]\}$ to $\mathit{MaxN_i}$ at the second pulse. After the second pulse, variable $\mathit{MaxN_i}$ is never changed for any $v_i$.

After the second pulse, $v_i$ enters a critical section once in $\mathit{MaxN_i}+1$ consecutive pulses. Hence, global fairness property is satisfied. For any $v_j\in N[i]$, since $\mathit{MaxN_j} \ge |N[i]|$ holds, $v_j$ enters a critical section at most once in $|N[i]|+1$ consecutive pulses. Hence, during $|N[i]|+1$ consecutive pulses, there is a configuration such that no process $v_j\in N[i]$ enters a critical section from the pigeon holes principle. Hence, local rendezvous property is satisfied. Since $l$-exclusion is always satisfied in case of $l=\Delta+1$, the theorem holds.
\end{proof}

\section{Preliminaries for Robot Systems}

In the following sections, we show the effectiveness of neighborhood mutual remainder by applying it to an implementation of LCM synchronization in mobile robot systems.

\subsection{Underlying robot model}

In the robot system, $k$ mobile robots exist in a plane. Robots do not know the value of $k$. Robots are anonymous and identical, that is, they have no unique identifiers and execute the same deterministic algorithm. Each robot has a memory. Each robot has a light, which can emit a color to other robots. Each robot $r$ can read information (i.e., positions and colors) of robots within a fixed distance from its current position. Robots have no direct communication means except for lights. A communication graph is defined as $G=(V,E)$ where $V$ is a set of robots and $E$ is a set of robot pairs that can read each other. Note that the communication graph may change when robots move. We say robot $r_i$ is a neighbor of $r_j$ if $(r_i,r_j)\in E$ holds. A neighborhood of robot $r_i$ is denoted by $N(i)=\{r_j \mid (r_i,r_j) \in E\}$, and a closed neighborhood of $r_i$ is denoted by $N[i]=N(i) \cup \{r_i\}$. 

\medskip

\noindent{\bf Dynamic graph reduction.}
Let $\phi$ be the distance a robot views, namely, the local neighborhood remainder algorithm of a robot is executed with all robots within $\phi$ distance from the robot. We assume that each robot moves up to $y<\phi$ in a single time unit and uses neighbors up to $\phi-y$ when executing \textsc{LOOK} and \textsc{COMPUTE}. Since $r_i$ does not execute \textsc{MOVE} when it executes \textsc{LOOK}, the only neighborhood dynamism is from another robot $r_j$ that is not viewed by $r_i$ in the neighborhood remainder algorithm, therefore is not included in the local synchronization, but penetrates to be in the \textsc{COMPUTE} zone of $r_i$ while $r_i$ executes \textsc{LOOK}. Hence having a $y$-tier eliminates such a scenario by any $r_j$.

\medskip

\noindent
{\bf Clock pulses.}
Robots operate based on pulses, which are generated in a partially-synchronous manner. When a robot receives a pulse, it instantaneously takes a snapshot by reading positions and colors of neighboring robots, and then computes and moves based on the snapshot before the next pulse. We consider two different pulses depending on the partial-synchronization assumptions. 
\begin{itemize}
\item \emph{Global pulses} are external pulses.
All robots receive these pulses at the same time, and the duration between two successive pulses is identical for all robots.
\item \emph{Local pulses} are generated locally.
All robots receive these pulses at different times, but the duration between two successive pulses is identical for all robots.
\end{itemize}
We regard the duration between two successive pulses as one time unit.

\subsection{LCM synchronization}

In this subsection, we describe traditional LCM (Look-Compute-Move) synchronization models we will implement on the underlying robot model. In the LCM model, each robot repeats three-phase cycles: \textsc{LOOK}, \textsc{COMPUTE}, and \textsc{MOVE}. During the \textsc{LOOK} phase, the robot looks positions and colors of its neighboring robots. During the \textsc{COMPUTE} phase, the robot computes its next state, color, and movement according to the observation in the \textsc{LOOK} phase. The robot may change its state and color at the end of the \textsc{COMPUTE} phase. If the robot decides to move, it moves toward the target position during the \textsc{MOVE} phase. The robot may stop moving before arriving at the target position, however it precedes the distance of at least $\sigma$. If the distance from the current position to the target position is at most $\sigma$, the robot always reaches the target position. In the following, we simply describe that a robot executes \textsc{LOOK}, \textsc{COMPUTE}, and \textsc{MOVE} instead of executing the \textsc{LOOK}, \textsc{COMPUTE}, and \textsc{MOVE} phase, respectively.

In literature, some synchronization models are considered in the LCM model. In this paper, we focus on two types of synchronicity: the move-atomic model and the FSYNC (fully-synchronous) model. 
The move-atomic model guarantees that, while a robot executes \textsc{MOVE}, none of its neighbors executes \textsc{LOOK} or \textsc{COMPUTE}. The FSYNC model guarantees that all robots synchronously execute \textsc{LOOK}, \textsc{COMPUTE}, and \textsc{MOVE}.

\subsection{Self-stabilizing LCM implementations}

We aim to implement the move-atomic model and the FSYNC model on the underlying system model in a self-stabilizing manner. To do this, we assign some time units to execute \textsc{LOOK}, \textsc{COMPUTE}, and \textsc{MOVE}, and trigger the phases upon pulses. We assume that each phase does not last beyond the next pulse. This implies that the duration from a pulse to the next pulse is sufficiently long so that robots precede the distance of at least $\sigma$.
 
\begin{definition}
The system implements a self-stabilizing move-atomic model if there exists some time $t$ such that, after time $t$, (1) every robot repeats three-phase cycles infinitely and 
(2) while a robot executes \textsc{MOVE}, none of its neighbors executes \textsc{LOOK} or \textsc{COMPUTE}.
\end{definition}

\begin{definition}
The system implements a self-stabilizing FSYNC model if there exists some time $t$ such that 1) every robot repeats three-phase cycles infinitely after time $t$ and 2) the time period $[t,\infty]$ is divided into infinitely many \textsc{LOOK}, \textsc{COMPUTE}, and \textsc{MOVE} periods such that, in each \textsc{LOOK} (resp., \textsc{COMPUTE} and \textsc{MOVE}) period, every robot executes \textsc{LOOK} (resp., \textsc{COMPUTE} and \textsc{MOVE}) exactly once and does not execute any other phase.
\end{definition}

\section{Self-Stabilizing Move-Atomic Algorithm with Global Pulses}

In this section, we consider an implementation of self-stabilizing move-atomic model, where we assume there is an external clock for global pulses.

The main idea of the implementation is to apply the neighborhood mutual remainder algorithm in Section \ref{sec:nmr} to robots. We allow robot $r_i$ to execute LOOK and COMPUTE only when $r_i$ enters a critical section, and allow $r_i$ to execute MOVE only when the neighborhood of $r_i$ is in rendezvous, namely, $rendezvous_i$. When $rendezvous_i$ takes place, no neighbor of $r_i$ is in the critical section (\emph{i.e.}, no neighbor of $r_i$ executes LOOK or COMPUTE), and thus $r_i$ can execute MOVE. By this behavior, we can achieve the move-atomic property: While a robot executes MOVE, none of its neighbors executes LOOK or COMPUTE.

Let $k$ be the number of robots.
Each robot $r_i$ has following two lights:
\begin{itemize}
\item $\mathit{Nlight}_i\in\{1,\dots,k\}$: the color represents $|N[i]|$. 
\item $\mathit{Light}_i\in\{0,\dots,k\}$: the color represents the value of the local clock phase based on global pulses. 
\end{itemize}
Additionally, $r_i$ maintains the following variables:
\begin{itemize}
\item $\mathit{MaxN}_i\in\{1,\dots,k\}$: the maximum value of $\mathit{Nlight}$ among the closed neighborhood of $r_i$. 
\item $\mathit{LC}_i$: a Boolean which represents whether the next operation is \textsc{LOOK} or not. 
\item $\mathit{Clock}_i\in\{0,\dots,k\}$: 
a local counter of the global pulses, not necessarily identical value among the robots.
\end{itemize}

\begin{algorithm}[htb] 
\begin{tabbing}
10: \= sp \= sp \= sp \= sp \= \kill
\XLN
\LN \> {\bf Upon} a global pulse\\
\LN \> \> ${\it MaxN}_i:=\max\{{\it Nlight}_j~|~r_j\in N[i]\}$\\
\LN \> \> {\bf if} $\forall r_j \in N[i][{\it Light}_j \neq 0]\land {\it LC}_i = {\it false}$ {\bf then}$\{$\\
\LN \> \> \> // Rendezvous (No closed neighbors enter a critical section)\\
\LN \> \> \> execute {\sc MOVE}\\
\LN \> \> \>${\it Nlight}_i$:=$|N[i]|$\\
\LN \> \> \>${\it LC}_i:={\it true}$\\
\LN \> \> $\}${\bf else} {\bf if} ${\it Light}_i = 0 \land {\it LC}_i = {\it true}$ {\bf then}$\{$ \\
\LN \> \> \> // Enter a critical section\\
\LN \> \> \> execute {\sc LOOK}\\
\LN \> \> \> execute {\sc COMPUTE}\\
\LN \> \>  \> ${\it LC}_i:={\it false}$\\
\LN \> \>  $\}$\\
\LN \> \> $\mathit{Clock}_i := (\mathit{Clock}_i+1) \bmod (\mathit{MaxN}_i+1)$ \\
\LN \> \> ${\it Light}_i$ := ${\it Clock}_i$
\end{tabbing}
\caption{Self-Stabilizing Move-Atomic Algorithm with Global Pulses for $r_i$}\label{ccp}
\end{algorithm}

When $r_i$ detects a global pulse, $r_i$ obtains visible neighbors' $\mathit{Nlight}$ values and updates $\mathit{MaxN}_i$.
The local counter of global pulses $\mathit{Clock}_i$ is bounded by $\mathit{MaxN}_i$, and maintained by each robot $r_i$, that is, they are not necessarily the same.
By the value of its counter, each robot decides its color of $\mathit{Light}_i$. 
When $\mathit{Light}_i$ is $0$, $r_i$ can execute \textsc{LOOK} and \textsc{COMPUTE} (\emph{i.e.}, $r_i$ enters a critical section). 
Only immediately after all values of $\mathit{Light}$ of its closed neighbors become not $0$, meaning none are planning to execute \textsc{LOOK} and \textsc{COMPUTE} in the next (long) global pulse, it can execute \textsc{MOVE} (\emph{i.e.}, no closed neighbors enter a critical section and hence rendezvous is satisfied). 
Then, because the visible graph changes, $|N[i]|$ also changes.
Thus, after a \textsc{MOVE} execution, $r_i$ updates the color of $\mathit{Nlight}_i$.

Because the $\mathit{Light}_i$ value is $0,\ldots,\mathit{MaxN}_i$, even if the neighbors and itself have different light value, there is a time when no light is $0$ value among the neighbor and itself (See Algorithm~\ref{ccp}).  
The time means none are planning to execute \textsc{LOOK} in the next pulse.
In the next pulse, no one has $1$ clock value.
Thus, if it has not executed \textsc{MOVE} yet, it can execute \textsc{MOVE} of the original LCM algorithm when all are non zero.

\begin{figure}[h]
\begin{minipage}[t]{\hsize}
\centering\includegraphics[keepaspectratio, width=13cm]{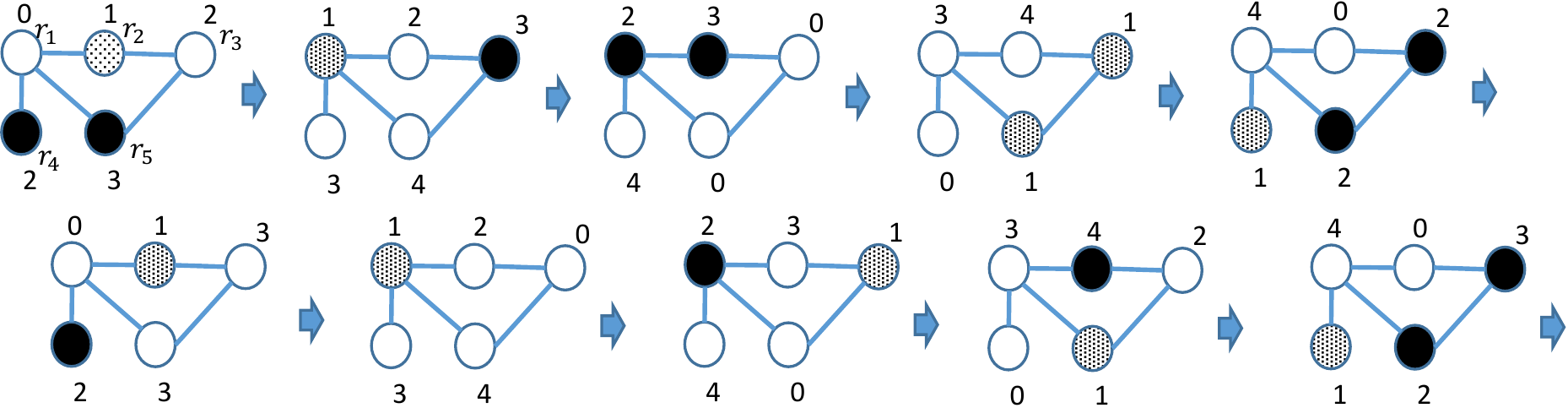}
\caption{There is a time when no light is $1$ value among its neighbor and itself.}\label{Figure1}
\end{minipage}\\
\begin{minipage}[t]{\hsize}
\begin{center}
\centering\includegraphics[keepaspectratio, width=5cm]{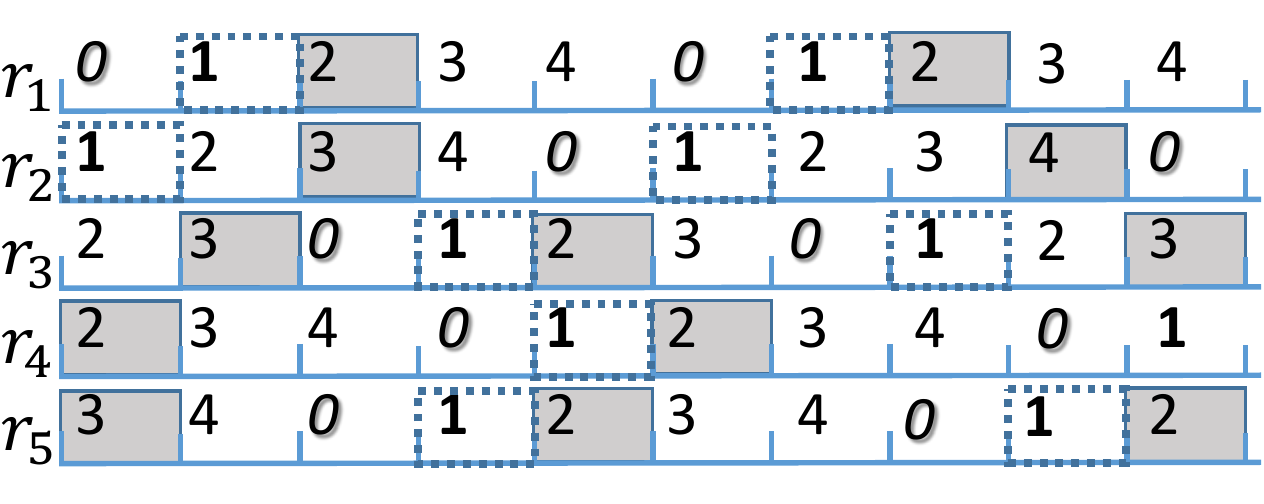}
\caption{Time Diagram for Self-Stabilizing Move-Atomic with a Global Pulses}
\label{tlccp}
\end{center}
\end{minipage}
\end{figure} 

Figures~\ref{Figure1} and \ref{tlccp} describe the same execution in two different ways. The clock value is the value in the end of the round.
Figure~\ref{Figure1} follows the policy of Algorithm 2 which allows a robot to execute \textsc{MOVE} if none of the neighbors have light $1$. 
Black (resp. Gray) nodes represent that the robot can \textsc{MOVE} (resp. \textsc{LOOK} and \textsc{COMPUTE}).
Figure~\ref{tlccp} takes the compliment policy allowing a robot to execute \textsc{LOOK} and \textsc{COMPUTE} when none of its neighbors execute \textsc{MOVE}. 
Any time slot in Figure~\ref{tlccp} with a solid (resp. dashed) box allows to execute \textsc{MOVE} (resp. \textsc{LOOK} and \textsc{COMPUTE}) by the algorithm.
The robots may use only one light, light 0, as the absent of a neighboring robot with light 0 indicates that no neighbor will be in 1 in the next configuration. 

\begin{lemma}\label{LA}
Eventually whenever a robot executes \textsc{MOVE} no other robot executes \textsc{LOOK} and \textsc{COMPUTE} and vise versa.
\end{lemma}
\begin{proof}
Each robot changes its light to 0 prior to executing \textsc{LOOK} (and \textsc{COMPUTE}) and no robot executes \textsc{MOVE} in the time period that follows another robot detects 0 value lights.
\end{proof}

\begin{lemma}\label{LB}
Each robot executes \textsc{LOOK}, \textsc{COMPUTE}, and \textsc{MOVE} infinitely often.
\end{lemma}
\begin{proof}
Consider an arbitrary robot $r_i$.
$r_i$ changes its light to 0 in every $\mathit{MaxN}_i+1$ successive rounds, and the neighbors of $r_i$ change their light to 0 at most once in every $\mathit{MaxN}_i+1$ successive rounds.
Since $|N[i]| \leq \mathit{MaxN}_i$, there must be a round in which no neighbor of $r_i$ changes its light to 0, then in the next round $r_i$ can execute \textsc{LOOK} and \textsc{COMPUTE}. 
In all rounds that follows $r_i$ changes its light to 0, $r_i$ can execute \textsc{MOVE}.
\end{proof}

By Lemmas~\ref{LA} and \ref{LB}, we derive the following theorem.

\begin{theorem}
Algorithm 2 implements the self-stabilizing move-atomic model. 
\end{theorem}

\section{Self-Stabilizing Move-Atomic Algorithm with Local Pulses}

Now, we consider the case that there is no global pulse.
Then, each robot executes the algorithm based only on local pulses.

We assume that the duration of local pulses are the same for each robot.
Let $\mathit{Lclock}_i$ be a local pulse counter for a robot $r_i$.
However, they are not ticking together. 
When pulses are not synchronized they can be slightly less that one time unit apart, thus we triple the time reserved for \textsc{MOVE}, and \textsc{MOVE} is executed only in the middle time unit of these three. So if the light of a neighbor is 0, 1 or 2, we regard it as the original zero and \textsc{MOVE} is executed when our clock satisfies $(\mathit{Lclock}_i \bmod 3) = 1$. So if no neighbor clock is in $[0..2]$, we are ready to execute \textsc{MOVE} but wait until our clock is the next middle one.
To this end, in Algorithm 3, we triple the value of clock $\mathit{Light}_i$, that is, $\mathit{Light}_i$ := $(\lfloor\mathit{Lclock}_i/3\rfloor).(\mathit{Lclock}_i \bmod 3)$.

Each robot $r_i$ has following two lights:
\begin{itemize}
\item $\mathit{Nlight}_i\in\{1,\dots,k\}$: the color represents $|N[i]|$. 
\item $\mathit{Light}_i\in\{0,\dots,k\}$: the color represents the value of the local clock based on local pulses. 
\end{itemize}
Additionally, $r_i$ maintains the following variables:
\begin{itemize}
\item $\mathit{MaxN}_i\in\{1,\dots,k\}$: the maximum value of $\mathit{Nlight}$ among the closed neighborhood of $r_i$. 
\item $\mathit{LC}_i$: a Boolean which represents whether the next operation is \textsc{LOOK} or not.
\item $\mathit{Lclock}_i\in\{0,\dots,k\}$: 
a local counter of the local pulses, not necessarily identical value among the robots.
\end{itemize}

\begin{algorithm}[htb] 
\begin{tabbing}
 10: \= sp \= sp \= sp \= sp \= \kill
\XLN
\LN \> {\bf Upon} a local pulse\\
\LN \> \> ${\it MaxN}_i:=\max\{{\it Nlight}_j~|~r_j\in N[i]\}$\\
\LN \> \>{\bf if} $\forall r_j \in N[i][{\it Light}_j \not\in \{0.2,1.0,1.1\}] \land ({\it Lclock}_i\bmod 3) = 1 \land {\it LC}_i = {\it false} ~{\bf then}\{$\\
\LN \> \> \> execute {\sc MOVE}\\
\LN \> \> \> ${\it Nlight}_i$:=|N[i]|\\
\LN \> \> \> ${\it LC}_i$:= {\it true}\\
\LN \> \>$\}${\bf else} ${\it Light}_i = 1.0 \land {\it LC}_i = {\it true}$ ~{\bf then}$\{$\\
\LN \> \> \> execute {\sc LOOK}\\
\LN \> \> \> execute {\sc COMPUTE}\\
\LN \> \> \> ${\it LC}_i$:= {\it false}\\
\LN \> \> $\}$\\
\LN \> \>${\it Lclock}_i$:= (${\it Lclock}_i+1) \bmod (3{\it MaxN}_i+3)$ \\
\LN \> \>${\it Light}_i$ := $(\lfloor{\it Lclock}_i/3\rfloor).({\it Lclock}_i \bmod 3)$
\end{tabbing}
\caption{Self-Stabilizing Move-Atomic Algorithm with Local Pulses for $r_i$}\label{icp}
\end{algorithm}
\begin{figure}[t]
\begin{minipage}[t]{\hsize}
\centering\includegraphics[keepaspectratio, width=13cm]{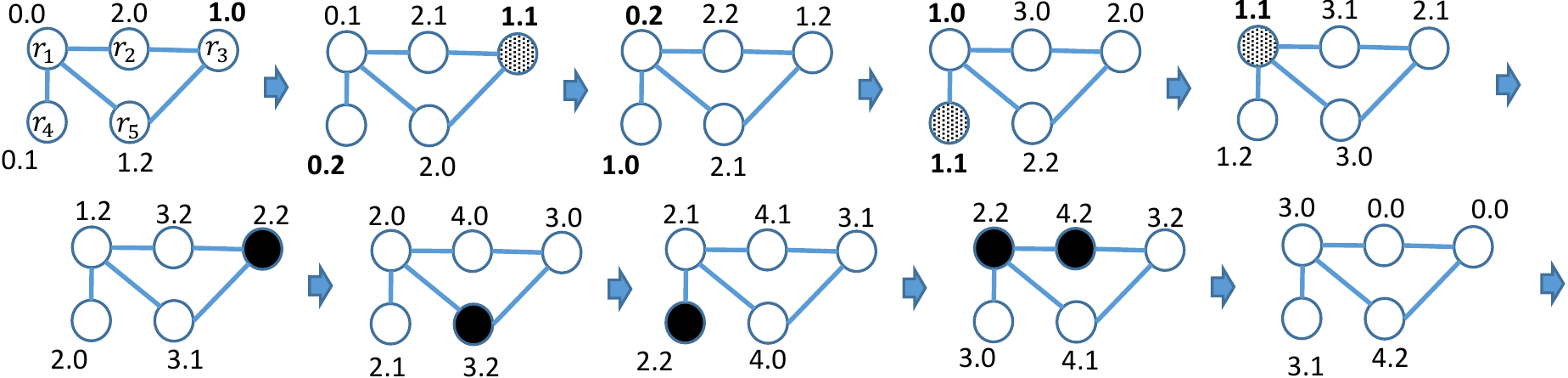}
\caption{There is a time when no light is $1$ value among its neighbor and itself.}\label{Figure2}
\end{minipage}\\
\begin{minipage}[t]{\hsize}
\centering\includegraphics[keepaspectratio, width=8cm]{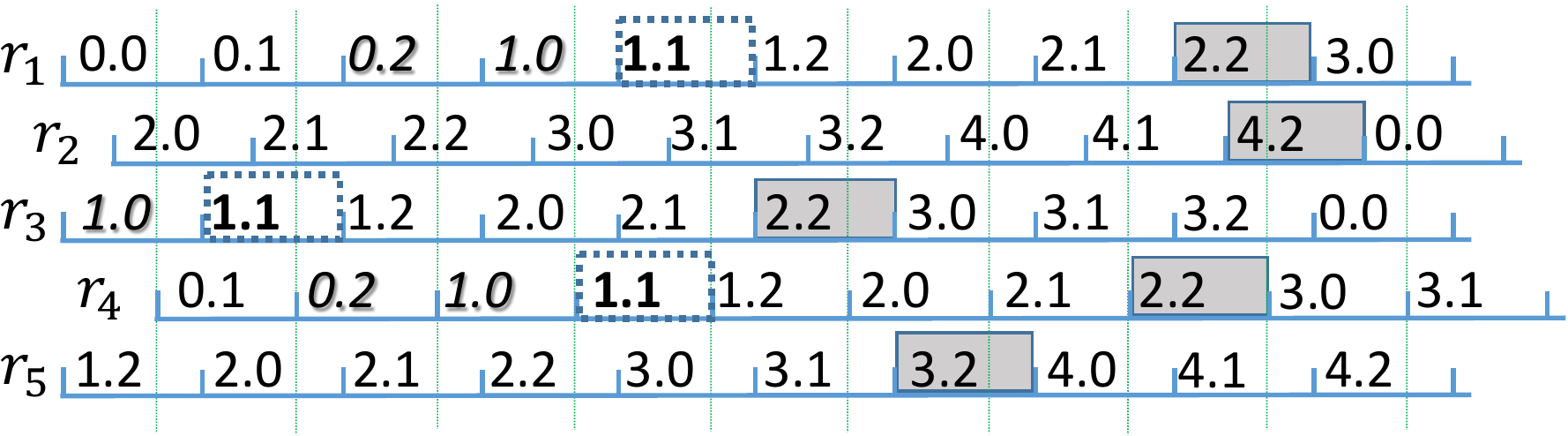}
\caption{Time Diagram for Self-Stabilizing Move-Atomic with Local Pulses}
\label{tlicp}
\end{minipage}
\end{figure} 

Figures \ref{Figure2} and \ref{tlicp} describe the same execution in two different ways.
The clock value is the value in the end of the round.
Figure \ref{tlicp} demonstrates the case that the timing of local pulses is off. 
In the round which $\mathit{Light}_i$ becomes 1.1, if $r_i$ observed one of 0.2, 1.0 and 1.1 lights, $r_i$ executes \textsc{LOOK} and \textsc{COMPUTE}.
If $r_i$ did not observe 0.2, 1.0 and 1.1 lights and $(\mathit{Lclock}_i \bmod 3)$ is 1, $r_i$ executes \textsc{MOVE}. Note that, then $(\mathit{Lclock}_i \bmod 3)$ becomes 2 in the round.

\begin{lemma}\label{LA1}
Eventually whenever a robot executes \textsc{MOVE}, no other robot in the closed neighborhood executes \textsc{LOOK} and \textsc{COMPUTE} and vice versa.
\end{lemma}
\begin{proof}
Assume towards contradiction that $r_i$ executes \textsc{LOOK} while $r_j\in N[i]$ executes \textsc{MOVE}. Before $r_i$ executes \textsc{LOOK} the light of $r_i$ is 0.2 and 1.0 for two subsequent time units, $r_j$ must find this value and hence does not execute \textsc{MOVE} when $r_i$ executes \textsc{LOOK}. 
Just after \textsc{LOOK} operations of $r_i$, if $r_i$'s right is changed to 1.1,
$r_j$ does not execute \textsc{MOVE} when $r_i$ executes \textsc{LOOK}.
\end{proof}

\begin{lemma}\label{LB1}
Each robot executes \textsc{LOOK}, \textsc{COMPUTE} and \textsc{MOVE} infinitely often.
\end{lemma}
\begin{proof}
Consider an arbitrary robot $r_i$.
$r_i$ changes its light value to 0.2, 1.0 or 1.1 in every $3\mathit{MaxN}_i+3$ successive rounds.
The neighbors of $r_i$ change their light at most once in every $3\mathit{MaxN}_i+3$ successive rounds.
Since $\mathit{MaxN}_i \geq |N[i]|$, there must be a round in which no neighbor of $r_i$ changes its light to 0.2, 1.0 or 1.1.
Then, in the round in which $r_i$'s light becomes 1.1, $r_i$ can execute \textsc{LOOK}. 
In all rounds that follows $r_i$ changes its light to the value such that $(\mathit{Lclock}_i \bmod 3)=1$ holds, $r_i$ can execute \textsc{MOVE}.
\end{proof}

By Lemmas~\ref{LA1} and \ref{LB1}, we derive the following theorem.

\begin{theorem}
Algorithm 3 implements the self-stabilizing move-atomic model. 
\end{theorem}

\section{Self-Stabilizing FSYNC Algorithm with Local Pulses}

We propose a novel local pulses version Algorithm~\ref{glcm} of the min clock selection of \cite{ADG91}.
Let $D$ be a diameter of the communication graph. Here we assume the \emph{regular register} semantics (see e.g., \cite{L96,D2K,AW04}) where 
a read segment that overlaps light changes may return one of the overlapping lights (rather than all as we assumed previously).
Thus, when a neighbor changes a light from 0 to 1 while a robot reads it can observe either 0 or 1. In case all participants are connected in a line, the above semantic can cause a difference of 1 clock value between any two neighbors, and $D=n$ from side to side.  
\begin{algorithm}[tbh]
\begin{tabbing}
10: \= sp \= sp \= sp \= sp \= \kill
\XLN
\LN \> \textbf{Upon} a local pulse\\
\LN \> \>$\mathit{Light}_i:= (min\{\mathit{Light}_j ~|~ r_j \in N[i]\}+1) \bmod (6D+1)$\\
\LN \> \> \textbf{if} $\mathit{Light}_i = 2D$ \textbf{then} $\{$\\
\LN \> \> \> execute \textsc{LOOK}\\
\LN \> \> \> execute \textsc{COMPUTE}\\
\LN \> \> $\}$ \textbf{else if} $\mathit{Light}_i = 4D$ \textbf{then} $\{$\\
\LN \> \> \> execute \textsc{MOVE}\\
\LN \> \> $\}$
\end{tabbing}
\vspace{-0.3cm}
\caption{Self-Stabilizing FSYNC Algorithm for $r_i$}\label{glcm}
\end{algorithm}

For ease of description we use $6D+1$ light values, in fact the number of lights may be close to $2D+1$ as in the case of a global clock pulse suggested in \cite{ADG91}.
The fact that pulses are according to local clocks require a change of the stabilization proof of \cite{ADG91}.
Still it holds that during $6D+1$ local pulses of all participants, at least one participant, say $r_i$, assigns 0 to $\mathit{Light}_i$,
and in the next $D$ local pulses of every participants, all the lights will be less than $D+1$.
Moreover in the next $D$ local pulses of all robots, the lights will be at most one value apart forever (where 0 follows $6D$). 
Thus all can execute \textsc{LOOK} and \textsc{COMPUTE} when their light is $2D$ and \textsc{MOVE} when their light is $4D$.
We assume that each robot reads all lights during the time unit (See. Figures \ref{Figure8} and \ref{Figure9}).

\begin{figure}[t]
\begin{minipage}[t]{\hsize}
\centering\includegraphics[keepaspectratio, width=13cm]{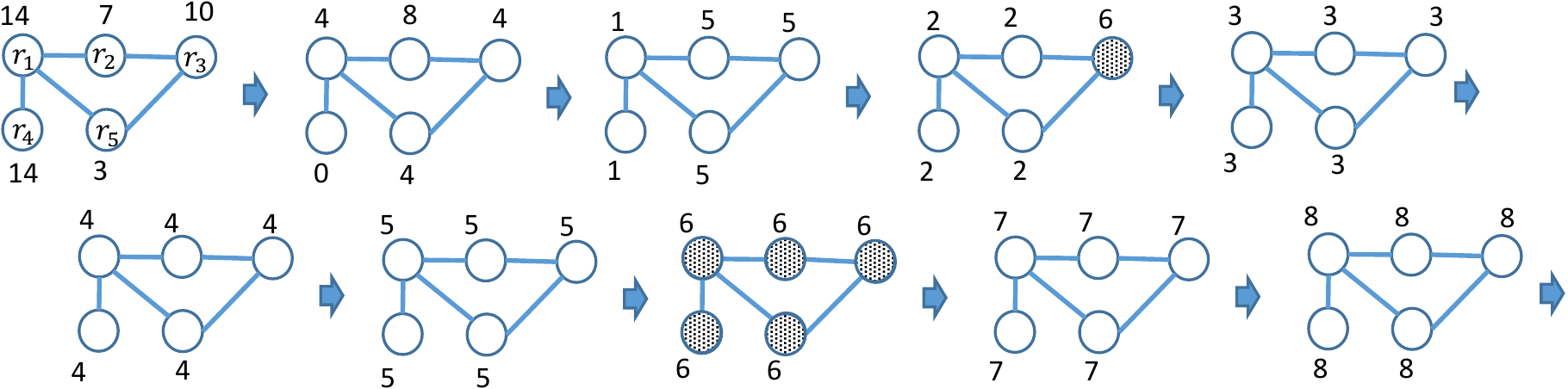}
\caption{Each robot reads all lights during the time unit.}\label{Figure8}
\end{minipage}\\ 
\begin{minipage}[t]{\hsize}
\begin{center}
\centering\includegraphics[keepaspectratio, width=8.5cm]{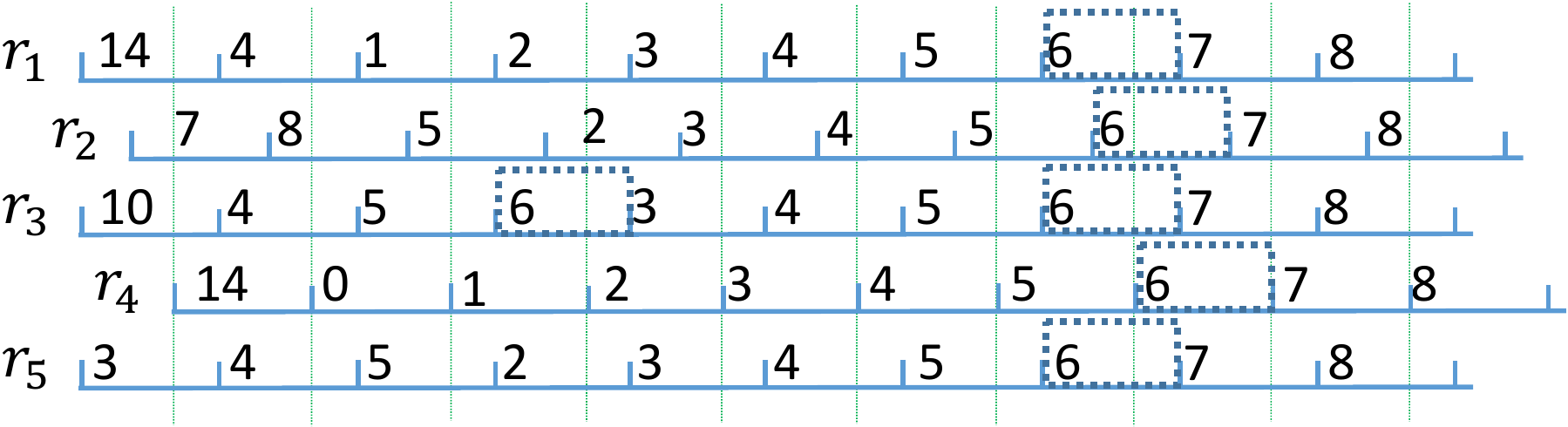}
\caption{Time Diagram for Self-Stabilizing FSYNC Algorithm, when each robot reads all lights during the time unit.}
\label{Figure9}
\end{center}
\end{minipage}
\end{figure} 

The case in which each robot reads one light at a time is straight-forward requiring a larger, by a factor related to the number of possible neighbors, clock value bound.    

\section{Conclusions}
Implementing algorithms in practice requires to rely on the abstracted away details that the algorithm designers make. We present practical implementation of LCM encountering new synchronization distributed computing task, i.e., the neighborhood mutual remainder.

Our results are described for the case of non-drifting local clocks,
where a time unit is identical across all clocks. Still, local clock speeds may slightly vary due to a constant bound on mutual clock drifts. 
Ensuring overlaps in (consecutive logical) time units in a fashion similar to the one suggested in the last two sections will cope with such clock drifts as well.

\subsubsection*{Acknowledgement}
This research was supported by Japan Science and Technology Agency (JST) as part of Strategic International Collaborative Research Program (SICORP).

\end{document}